\documentclass[11pt]{article}
\usepackage[totalwidth=13.0cm,totalheight=20.0cm]{geometry}
\usepackage{latexsym,amsthm,amsmath,amssymb,url}
\usepackage[ruled, linesnumbered]{algorithm2e}

\usepackage{tikz,authblk}
\usetikzlibrary{decorations.pathreplacing}

\newtheorem{lemma}{Lemma}

\newtheorem{definition}{Definition}
\newtheorem{theorem}{Theorem}

\begin{document}

\title{Note on Perfect Forests in Digraphs}
\author[1]{Gregory Gutin}
\author[2,3]{Anders Yeo}
\affil[1]{Department of Computer Science, Royal Holloway, University of London, TW20 0EX, Egham, Surrey, UK}
\affil[2]{Engineering Systems and Design, Singapore University of Technology and Design, 8 Somapah Road 487372, Singapore} 
\affil[3]{Department of Mathematics, University of Johannesburg, Auckland Park, 2006 South Africa}

\date{}
\maketitle

\begin{abstract}
\noindent A spanning subgraph $F$ of a graph $G$ is called {\em perfect} if $F$ is a forest,
the degree $d_F(x)$ of each vertex $x$ in $F$ is odd, and
each tree of $F$ is an induced subgraph of $G$.
Alex Scott (Graphs \& Combin., 2001) proved that every connected graph $G$ contains a perfect forest if and only if $G$ has an even number of vertices. 
We consider four generalizations to directed graphs of the concept of a perfect forest. While the problem of existence of the most straightforward one is NP-hard, for the three others this problem is polynomial-time solvable. Moreover, every digraph with only one strong component contains a directed forest of each of these three generalization types. One of our results extends Scott's theorem to digraphs in a non-trivial way.
\end{abstract}

\pagestyle{plain}

\section{Introduction}\label{sec:intro}

The number of vertices of a graph $G$ is called its {\em order}. A spanning subgraph $F$ of a graph $G$ is called a {\em perfect forest} if $F$ is a
forest, the degree $d_F(x)$ of each vertex $x$ in $F$ is odd, and each tree of $F$
is an induced subgraph of $G$. Note that a perfect matching is a perfect forest. Clearly, if a connected graph $G$ has a perfect forest, then $G$ is of even order.
Alex Scott \cite{Sco2001} proved that surprisingly the opposite implication is also true.

\begin{theorem} \label{undirected} \cite{Sco2001}
A connected graph $G$ contains a perfect forest if and only if $G$ has an even number of vertices.
\end{theorem}

Intuitively, it is clear that a perfect forest can provide a useful structure in a graph and, in particular, this notion was used by Sharan and Wigderson \cite{SW} to prove that the perfect matching problem for bipartite cubic graphs belongs to the complexity class ${\cal NC}$.

Scott's original proof was graph-theoretical; a shorter one based on linear algebra over $\mathbb{F}_2$ was recently given in \cite{GG}. Recently, Yair Caro \cite{Caro} produced another shorter proof, which is graph-theoretical.

I this paper, we will generalize Scott's theorem and related notions to directed graphs.\footnote{All basic notions on directed graphs not introduced in this paper can be found in the monograph \cite{JBJGG}.} Often there is no single generalization of an undirected graph notion to directed graphs. For example, connectivity in undirected graphs can be generalised to {\em strong connectivity} when there are directed paths in both directions for every pair of vertices, {\em connectivity} when the {\em underlying graph} (i.e., the undirected graph obtained by replacing all arcs with edges with the same end-vertices and replacing all parallel edges by single edges) is connected, and {\em unilateral connectivity} when there is a directed path between every pair of vertices but possibly only in one direction. 
Another example is the notion of a tree. We can consider any orientations of trees, but often only out-trees (or in-trees) are of interest: an orientation of a tree is called an  {\em out-tree} if it has only one vertex of in-degree zero called its {\em root} and all other vertices are of in-degree one. An {\em out-forest} in a directed graph is a vertex-disjoint collection of out-trees. Let us introduce the following four generalizations of the notion of a perfect forest. 

\begin{definition}\label{def1}
A spanning subgraph $F$ of a digraph $D$ is called a {\em perfect out-forest} if $F$ is an
out-forest, the degree of each vertex in the underlying graph of $F$ is odd, and each out-tree of $F$
is an induced subgraph of $D$. 
\end{definition}

\begin{definition}\label{def2}
A spanning subgraph $F$ of a digraph $D$ is called an {\em almost perfect out-forest} if $F$ is an
out-forest, the degree of each vertex in the underlying graph of $F$ is odd, and every arc in $D$ is either in $F$, goes between different out-trees in $F$
or goes from a vertex to an ancestor of that vertex in an out-tree in $F$ (the last kind of arc is called a {\em backward arc}). 
%In other words, each out-tree in $F$ does not need to be induced, but any arc between vertices from the same out-tree either belongs to the out-tree or is a 
%which forms a directed cycle with arcs of the tree. 
\end{definition}

\begin{definition}\label{def3}
A spanning subgraph $F$ of a digraph $D$ is called a {\em weak perfect out-forest} if $F$ is an
out-forest and the degree of each vertex in the underlying graph of $F$ is odd.
\end{definition}

\begin{definition}\label{def4}
A spanning subgraph $F$ of a digraph $D$ is called a {\em even out-forest} if $F$ is an
out-forest where every out-tree in $F$ has even order.
\end{definition}

We will consider three families of connected digraphs (i.e., the underlying graphs are connected) of even order. 
The class of all strongly connected digraphs of even order will be denoted by ${\cal D}^{st}$. The class of all connected digraphs of even order
which contain only one initial strong component will be denoted by ${\cal D}^u$. Alternatively, any digraph in ${\cal D}^u$ contains
some vertex, say $u$, such that for any vertex in the digraph it can be reached from $u$ by a directed path.
The class of all connected digraphs of even order will be denoted by ${\cal D}$. 

In the next section, we will show several results on the generalizations of a perfect forest. 
We summarize the main results in the following table.

\begin{center}
\begin{tabular}{|c||c|c|c|} \hline
 &  ${\cal D}^{st}$ & ${\cal D}^u$ & ${\cal D}$ \\ \hline \hline
perfect    &   $NP$-hard   & $NP$-hard  & $NP$-hard   \\
out-forest &   to decide  & to decide & to decide  \\ \hline
almost perfect &  Always & Always  &  May not exist, but can be \\
out-forest     &  Exists &  Exists &  decided in polynomial time \\ \hline
weak perfect   &  Always & Always  & May not exist, but can be \\
out-forest     &  Exists &  Exists &   decided in  polynomial time \\ \hline
Even           &  Always & Always  & May not exist, but can be \\
out-forest     &  Exists &  Exists &   decided in polynomial time \\ \hline
\end{tabular}
\end{center}

The concept of a backward arc of an out-tree $T$ was introduced in Definition \ref{def2}. A {\em forward arc}  is not an arc of $T$, but it goes from an ancestor of a vertex to the vertex itself. Finally, any arc between two vertices of $T$ is called a {\em cross arc} if it is not on $T$, and not a forward or backward arc.  

Our result that every strongly connected digraph of even order contains an almost perfect out-forest, generalizes Scott's theorem. Indeed, let $G$ be a connected graph of even order and let $G^*$ be the digraph obtained from $G$ by replacing every edge $\{x,y\}$ with two arcs $xy$ and $yx$. Clearly, $G^*$ is strongly connected. By our result, $G^*$ has an almost perfect out-forest $F$. Note that $F$ has no backward arcs as every such arc would imply the existence of a forward arc, which is not allowed for almost perfect out-forests. 
Thus, the underlying graph of every out-tree in $F$ is an induced tree in $G$ and so $F$ is simply an orientation of a perfect forest in $G$.

\section{Results}\label{sec:proofs}

Let us start with the following two lemmas which imply that that if a digraph $D \in {\cal D}$ has an even out-forest, then it also has a a weak perfect out-forest and even an almost perfect out-forest.

\begin{lemma} \label{even_to_weak}
If $D \in {\cal D}$ contains an even out-forest then $D$ contains a weak perfect out-forest.
\end{lemma}

\begin{proof}
It suffices to show that any out-tree of even order contains a weak perfect out-forest as a spanning subgraph, as we can use this 
on each out-tree in an even out-forest in order to get the desired weak perfect out-forest.
We will prove this by induction on the number of vertices in the out-tree.

Clearly the statement is true if the out-tree has order two, as then the out-tree is also a weak perfect out-forest.
So now assume that $T$ is an out-tree of even order with root $r$ and that all out-trees of smaller order than $T$ 
contain a weak perfect out-forest.  
Let $u$ be a vertex of $T$ of maximum distance from the root $r$ and let $v$ be the unique vertex with an arc into $u$ in $T$.
We will now consider the cases when $v$ has an arc to a leaf different to $u$ in $T$ and when it has no such arc.

If $v$ has an arc to a leaf, say $w$, different from $u$ in $T$ then let $T' = T \setminus \{u,w\}$.
Note that $T'$ is an out-tree rooted at $r$ as $u$ and $w$ are leaves of $T$.
By induction there is a weak perfect out-forest, $F'$, which is a spanning subgraph of $T'$. 
Adding the arcs $vu$ and $vw$ to the out-tree of $F'$ containing $v$ gives us a weak perfect out-forest which is a subgraph of $T$.

We may therefore assume that $v$ does not have an arc to a leaf different from $u$. As the distance from $r$ to $u$ was chosen to be maximum this
implies that $u$ is the only out-neighbor of $v$ in $T$. Let $T'' = T \setminus \{u,v\}$ and note that $T''$ is still an out-tree.
Therefore, by induction there is a weak perfect out-forest, $F''$, which is a spanning subgraph of $T''$.
Adding the out-tree  $vu$ to $F''$ gives us a weak perfect out-forest which is a spanning subgraph of $T$.
This completes the proof.
\end{proof}

\begin{lemma} \label{weak_to_almost}
If $D \in {\cal D}$ contains a weak perfect out-forest then $D$ contains an almost perfect out-forest.
\end{lemma}

\begin{proof}
Let $F$ be a weak perfect out-forest in $D$ and let $T$ be an arbitrary  out-tree in $F$.
Assume that there exists an arc $uv$ in $D$, which does not belong to $T$, but is either a forward arc of $T$ or a cross arc of $T$.
In the underlying graph of $T$ let $P$ be the unique path from $u$ to $v$. Let $F_T$ be obtained from $T$ by adding the arc $uv$ and
deleting all the arcs in $T$ which correspond to the edges of $P$. 

Note that all vertices in $F_T$  have maximum in-degree one as the edge on $P$ incident with $v$ corresponded to an arc of $T$ into $v$.
Furthermore in the underlying graph of $F_T$ all vertices still have odd degree as the degrees of $u$ and $v$ remain unchanged and all
other vertices in $P$ have their degree decreased by $2$.
Therefore $F_T$ is a weak perfect out-forest with fewer arcs than $T$.

We continue the above process for as long as possible. As each time the above operation is performed the number of arcs decreases it will terminate.
Furthermore it only terminates when there are no arcs, which don't belong to the weak perfect out-tree, which are either forward or cross arcs in some 
out-tree in the weak perfect out-forest. Therefore we have produced an almost perfect out-forest, which completes the proof.
\end{proof}

As an almost perfect out-forest is also an even out-forest
the two lemmas above imply the following:

\begin{theorem} \label{all_equal}
Every $D \in {\cal D}$ either contains all three of the below or none of the below.

\begin{itemize}
\item Almost perfect out-forest.
\item Weak perfect out-forest.
\item Even out-forest.
\end{itemize}
\end{theorem}

%\begin{proof}
%Let $D \in {\cal D}$. If $D$ contains an even out-forest then Theorem~\ref{even_to_weak} implies that $D$ contains a weak perfect out-forest.
%If $D$ contains a weak perfect out-forest then Theorem~\ref{weak_to_almost} implies that $D$ contains an almost perfect out-forest.
%Finally, if $D$ contains an almost perfect out-forest, $F$, then $F$ is also an even out-forest, as every graph has an even number of odd-degree vertices.
%Therefore if $D$ contains any of the three stated kind of out-forests, then it contains all.
%\end{proof}

Unfortunately, the problem of existence of a perfect out-forest in a strongly connected digraph is intractable.

\begin{theorem} \label{np_hard}
It is $NP$-hard to determine whether a strongly connected digraph has a perfect out-forest.
\end{theorem}

\begin{proof}
We will reduce from {\sc $3$-Dimensional Matching} whose instance $I$ is given by three disjoint sets $V_1,V_2,V_3$ of vertices and an edge set $E$ where each edge in $E$ contains exactly one vertex from each of the three sets. The question is whether there is a {\em 3-dimensional perfect matching}, which a collection of disjoint edges of $E$ containing all vertices of the three sets. 
We may assume that $|V_1|=|V_2|=|V_3|=k$ and let $|E|=m$.

We will now construct a strongly connected digraph, $D$ as follows.
Let $X$ be an independent set of $m-k$ vertices and let $Y$ be a set of $m$ vertices (corresponding to the $m$ edges in $E$).
Let $V(D)=V_1 \cup V_2 \cup V_3 \cup X \cup Y$ and add all arcs from $X$ to $Y$ and for each vertex in $Y$ add one arc to each of $V_1,V_2,V_3$ such that the out-neighborhood 
of each vertex in $Y$ corresponds to a distinct edge in $E$.
Finally add all possible arcs between $X$ and $V_1 \cup V_2 \cup V_3$ (in both directions) and add all possible arcs between the vertices of $Y$ (in both directions).
Note that $D$ is strongly connected as $D[X \cup V_1 \cup V_2 \cup V_3]$ is strongly connected (due to all the $2$-cycles) and every vertex in $Y$ has arcs into it from $X$ and arcs out of it to $V_1$ (and $V_2$ and
$V_3$).

We will show that $D$ contains a perfect out-forest if and only if $I$ contains a 3-dimensional perfect matching.

First assume that $I$ has a 3-dimensional perfect matching containing the edges in $M \subseteq E$. Let $Y_M \subseteq Y$ be the vertices of $Y$ corresponding to the edges in $M$.
Take all the out-trees with roots in $Y_M$ and one leaf from each $V_1$, $V_2$ and $V_3$. This leaves $m-k$ vertices in $X$ and $m-k$ vertices in $Y$ which can be matched up 
with a matching, which together with the out-trees form a perfect out-forest.

Now assume that $D$ has a perfect out-forest, $F$. No arcs between $X$ and $V_1 \cup V_2 \cup V_3$ can belong to $F$ as they all belong to $2$-cycles.
Therefore we must have $m-k$ out-trees in $F$ that have a vertex from $X$ as their root. As $D[Y]$ is a complete digraph each of these out-trees contain exactly one vertex from $Y$.
Due to the arcs between $X$ and $V_1 \cup V_2 \cup V_3$ no other vertices can belong to these $m-k$ out-trees.
Removing these $m-k$ out-trees leaves us with $k$ vertices from $Y$ and $k$ vertices in each of $V_1$, $V_2$ and $V_3$. Therefore $F$ must have $k$ out-trees containing
one vertex from $Y$ and one vertex from each of $V_1$, $V_2$ and $V_3$. These $k$ vertices from $Y$ therefore correspond to a 3-dimensional perfect matching in $I$.
This completes the proof.
\end{proof}

Unlike the problem in the previous theorem, the problem of existence of a weak perfect out-forest is tractable, as proved below. Thus, by Theorem~\ref{all_equal}, the problems of existence of an almost perfect out-forest and an even out-forest are also tractable.
\begin{theorem} \label{pol}
In polynomial time, we can decide whether a connected digraph has a weak perfect out-forest.
\end{theorem}

\begin{proof}
Let $D$ be any connected digraph. We may assume that $D$ has even order. Let $|V(D)|=n=2k$.
%If the order of $D$ is odd then $D$ does not contain a weak perfect out-forest, as a weak perfect out-forest has an even number of vertices.
%We can therefore assume that $D$ has even order. Let $|V(D)|=n=2k$.
We will create a new graph $G$ as follows.

For every vertex $u \in D$ let $X_u$ be a set of $2k-1$ vertices containing $k-1$ edges that form a matching and let $y_u$ be the isolated vertex in $X_u$.
For every arc $uv$ in $D$ add all edges from $X_u$ to $y_v$.
The resulting graph is the desired graph $G$.
We will show that $G$ has a perfect matching if and only if $D$ contains a weak perfect out-forest. As the problem of deciding whether a graph has a perfect matching, is polynomial-time solvable, this will complete the proof.

Assume that $G$ has a perfect matching, $M$.
Let $F_M$ contain all arcs $uv$ where one of the edges from $X_u$ to $y_v$ belong to $M$ (note that at most one such edge can belong to $M$).
$F_M$ has maximum in-degree one as if $uv$ and $wv$ belong to $F_M$ then two distinct edges of $M$ have $y_v$ as an end-point, a contradiction.
Also, if $j$ edges of $G[X_u]$ are in $M$ then $u$ is incident with $(2k-1)-2j$ arcs in $F_M$ which is clearly an odd number.

% The next few simple arguments will show that the claims in rows 2-4 and columns 2 and 3 of the table of Section \ref{sec:intro} do also hold.

If $F_M$ contains a directed cycle then we remove all arcs in such a cycle and note that the resulting subgraph also has maximum in-degree at most
one and every vertex is adjacent with an odd number of vertices. Removing arcs of directed cycles until no directed cycles remain, leave us with 
an acyclic digraph with no in-degree more than $1$ and all vertices incident to an odd number of arcs. This is precisely a weak perfect out-forest.

Now assume that we have a weak perfect out-forest, $F$, in $D$.  
For every arc, say $uv$, in $F$ we add an edge from $X_u$ to $y_v$ to the matching we are building. 
Furthermore we do this such that as few edges of each $X_w$ are covered. As
every vertex in $D$ is incident with an odd number of arcs in $F$ we note that an even number of vertices in each $X_w$ is uncovered.
We can now use the edges within each $X_w$ in order to get a perfect matching in $G$.
This completes the proof.
\end{proof}

The next few simple arguments will show that the claims in rows 2-4 of the table of Section \ref{sec:intro} also hold.

Note that not every connected digraph of even order has an even out-forest. Indeed, 
consider any star $K_{1,r}$ with an odd number of leaves, $r \geq 3$, and direct all edges towards the root. Clearly any out-forest will contain at least $r-1$ trees of order one and
will therefore not be an even out-forest.

%it is enough to consider a digraph $H$ on vertex set $X\cup Y\cup \{z\}$, where $X\cap Y=\emptyset$, $z\not\in X\cup Y$, and $|Y|=|X|-1$ and with arcs being all arcs from $X$ to $Y$ and all arcs from $Y$ to $z$. Clearly, at least one vertex in $X$ must form an out-tree by itself and so $H$ has no even out-forest.

However, if $D\in {\cal D}^u$ (and therefore also if $D \in {\cal D}^{st}$)  has an even number of vertices then, by Proposition 1.7.1 of \cite{JBJGG}, $D$ contains a spanning out-tree, which is clearly an even out-forest. Thus, by Theorem~\ref{all_equal}, $D$ also has a weak perfect out-forest and an almost perfect out-forest. 

\section*{Acknowledgement} Research of GG was partially supported by Royal Society Wolfson Research Merit Award.


\begin{thebibliography}{1}

\bibitem{JBJGG} J. Bang-Jensen and G. Gutin, Digraphs: Theory, Algorithms and Applications, 2nd Ed., Springer-Verlag, London, 2009.

\bibitem{Caro} Y. Caro, Private communication with G. Gutin, Aug 2015. 

\bibitem{GG} G. Gutin, Note on perfect forests, J. Graph Theory, doi: 10.1002/jgt.21897.

\bibitem{Sco2001} A.D. Scott, On induced subgraphs with all degrees odd,
 Graphs Comb. 17 (2001), 539--553.
 
 \bibitem{SW} R. Sharan and A. Wigderson, A new NC algorithm for perfect matching in bipartite cubic graphs,
4th Israeli Symposium on Theory of Computing and Systems (1996), pp. 202--207.

\end{thebibliography}
\end{document}